\documentclass[conference]{IEEEtran}
\usepackage{epsfig}
\usepackage{times}
\usepackage{float}
\usepackage{afterpage}
\usepackage{amsmath}
\usepackage{mathrsfs}
\usepackage{amstext}
\usepackage{amssymb,bm}
\usepackage{latexsym}
\usepackage{color}
\usepackage{graphicx}
\usepackage{amsmath}
\usepackage{amsthm}
\usepackage{graphicx}
\usepackage[center]{caption}
\usepackage{pstricks}
\usepackage{subfigure}
\usepackage{booktabs}
\usepackage{multicol}
\usepackage{lipsum}
\usepackage{dblfloatfix}
\usepackage{algorithm} 
\usepackage[noend]{algpseudocode} 
\usepackage[normalem]{ulem}
\usepackage{enumitem}
\usepackage{bbm}

\newcommand{\mcal}{\mathcal}
\newcommand{\msf}{\mathsf}

\newtheorem{thm}{Theorem}
\newtheorem{cor}[thm]{Corollary}
\newtheorem{lem}[thm]{Lemma}

\newtheorem{rem}{Remark}
\newtheorem{defin}{Definition}
\newtheorem*{defin*}{Definition}

\algnewcommand\algorithmicforeach{\textbf{for each}}

\algdef{S}[FOR]{ForEach}[1]{\algorithmicforeach\ #1\ \algorithmicdo}

\input{widebar_hack.tex}

\begin{document}

\title{Half-Duplex Routing is NP-hard} 
\author{
\IEEEauthorblockN{Yahya H. Ezzeldin$^\dagger$, Martina Cardone$^{\dagger}$, Christina Fragouli$^{\dagger}$, Daniela Tuninetti$^*$}
$^{\dagger}$ UCLA, Los Angeles, CA 90095, USA,
Email: \{yahya.ezzeldin, martina.cardone, christina.fragouli\}@ucla.edu\\
$^*$ University of Illinois at Chicago,
Chicago, IL 60607, USA, 
Email: danielat@uic.edu
}
\IEEEoverridecommandlockouts
\maketitle
\begin{abstract}
Routing is 
a widespread approach to transfer information from a source node to a destination node
in many 
deployed wireless ad-hoc networks.
Today's implemented routing algorithms seek to efficiently find the path/route with the largest Full-Duplex (FD) capacity, which is given by the minimum among the point-to-point link capacities in the path.
Such an approach may be suboptimal 
if then the nodes in the selected path are operated in Half-Duplex (HD) mode.
%
Recently, the 
capacity (
up to a constant gap that only depends on the number of nodes in the path) of an HD line network i.e., a path) has been 
shown to be equal to half of 
the minimum of the harmonic means of the capacities of two consecutive links in the path.
This paper asks the questions of whether it is possible to design a polynomial-time algorithm that efficiently finds the path with the largest HD capacity in a relay network.
{
This problem of finding that path} is shown to be NP-hard in general. 
However, if the number of cycles in the network is polynomial in the number of nodes, 
then a polynomial-time algorithm can indeed be designed.
\end{abstract}

\section{Introduction}

In recent years 
there have been promising advances 
in the design of Full-Duplex (FD) 
transceivers~\cite{duarte2014design,everett2016softnull}.
Proposed FD designs however require 
complex self-interference cancellation techniques.
Given this, in the near future it is envisioned that nodes will continue to operate in HD mode -- as recently announced, for example, in 3GPP~Rel-13~\cite{wang2017primer}.

Today's wireless ad-hoc networks 
route information from a source node to a destination node 
through a single multi-hop path, consisting of consecutive point-to-point links. 
Routing is considered an appealing option since a route/path can be efficiently operated, while providing energy savings (since only the nodes along the path are operated)
and limiting the level of interference in the network.
A rich body of literature on 
wireless routing exists~\cite{awerbuch2004high,de2005high,broch1998performance}. 
In this paper we 
revisit the problem of HD routing in wireless 
networks starting from the following observation.

   \begin{figure}[t]
        \centering
        \includegraphics[width=0.43\textwidth]{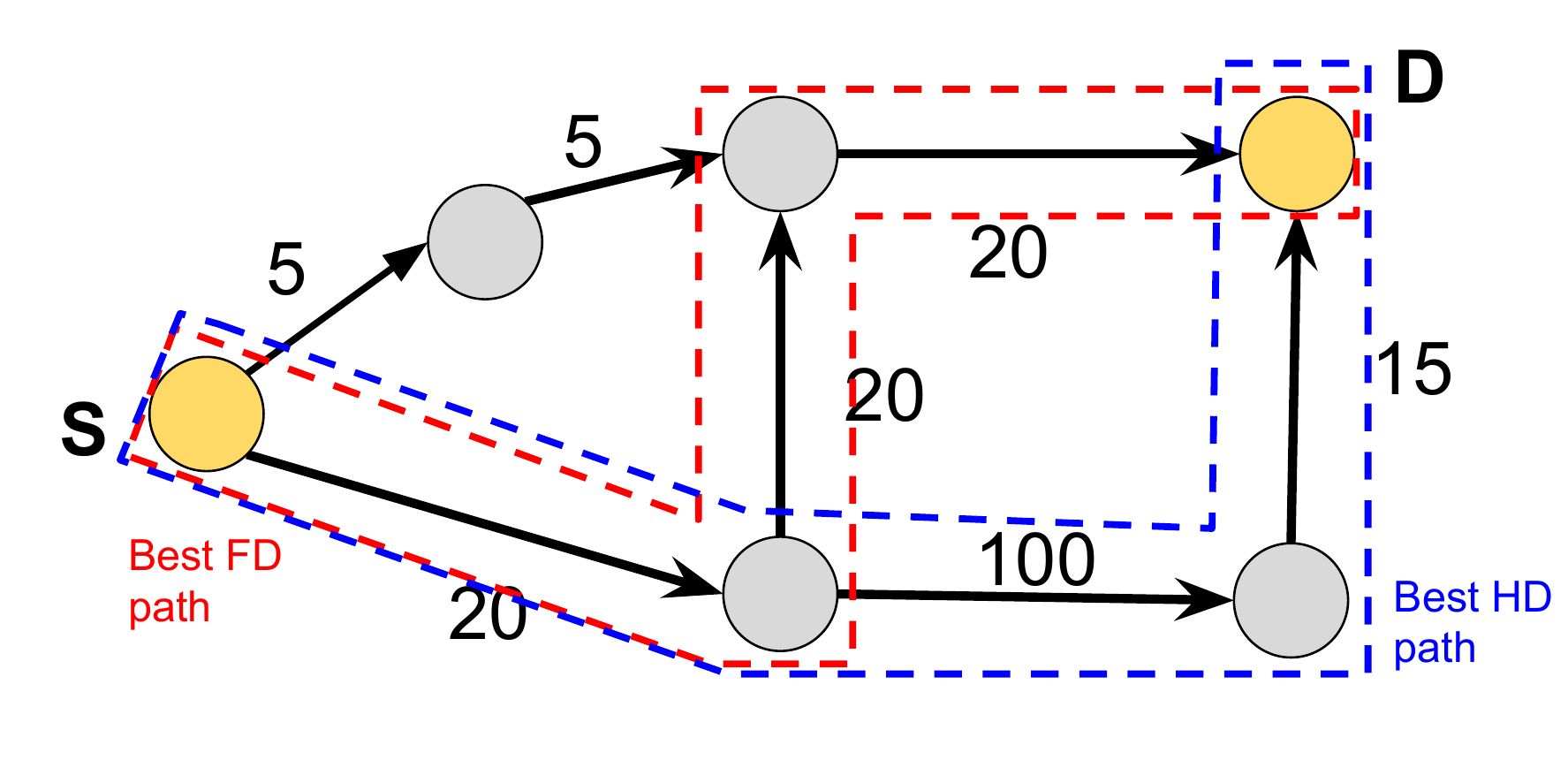}
        \caption{Example where the best FD and HD paths are different. Edge labels represent the point-to-point link capacities of the edges.}
        \label{fig:motivation_example}
    \end{figure}

Many routing algorithms seek to efficiently find the path/route with the largest FD capacity, which is given by the minimum among the point-to-point link capacities in the path.
Such an approach may be suboptimal 
if then the nodes in the selected path are operated in HD mode,
as shown by the example in Fig.~\ref{fig:motivation_example}.
Recall that, given a path that connects node $k_0$ to node $k_{N+1}$ through $N$ intermediate HD relays and point-to-point link capacities $c_{k_{i-1},k_i}$, $\forall i \in [1:N+1]$, the HD 
capacity (up to a constant gap that only depends on the number of nodes in the path) is given by~\cite{ARXIV_version} 
\begin{align}
    \msf{C} = \min_{i \in [1:N]} \left\{\frac{c_{k_{i-1},k_i}c_{k_{i},k_{i+1}}}{c_{k_{i-1},k_i} + c_{k_{i},k_{i+1}} }\right\}.
    \label{eq:cap_expressionNew}
\end{align}
%
By applying the expression in~\eqref{eq:cap_expressionNew}, we find that the best HD-route (within the blue box in Fig.~\ref{fig:motivation_example}) 
has an HD capacity of 13.04, which is 30\% higher than
the HD capacity of the best FD-route (within the red box in Fig.~\ref{fig:motivation_example}, with FD capacity of 20 but HD capacity of 10).
With best {FD (resp. HD)} route we refer to the path that has the largest {(FD resp. HD)} capacity.

The above example shows that in general it is suboptimal to find the best HD path by using as optimization metric the FD capacity of the path. 
In fact, one can show that there exist networks for which
routing based on the FD capacities yields a route 
with HD capacity equal to half 
that of the best HD route~\cite{ARXIV_submodularity}.
This observation 
naturally suggests the question: \emph{Does there exist an efficient (polynomial-time) algorithm that finds the route in a network with the best HD capacity?} 

\smallskip
\noindent{\bf{Contribution.}} 
The main result of this paper is 
 that the problem of finding the best HD route is NP-hard in general.
 Our proof is based on a reduction {from} the 3SAT problem~\cite{karp1972reducibility}.
Additionally, we provide a sufficient criterion -- based on the number of cycles in the network -- for the existence of a polynomial-time algorithm for some networks.


\smallskip
\noindent{\bf{Related Work.}} 
An HD route connecting a source node to a destination node through $N$ relays is an $N$-relay HD line network. 
For an HD line network, although the capacity is known to be given by the cut-set bound (since the line network is a degraded relay channel~\cite{aref1981information}), a closed-form expression of it as a function of the point-to-point link capacities is not {yet available}.
Recent results in~\cite{OzgurIT2013}, generalized an observation in~\cite{KramerAllerton2004}, by showing that the HD capacity of a Gaussian relay network can be approximated to within a constant gap (i.e., which is independent on the channel parameters and only depends on the number of relays $N$) by the cut-set upper bound evaluated with a deterministic schedule independent of the transmitted and received signals and with independent inputs. 
Throughout the paper, we refer to this as the \emph{approximate capacity}. 
Recently, in~\cite{ARXIV_version} we provided a closed-form characterization of the approximate capacity for a Gaussian HD line network (see the expression in~\eqref{eq:cap_expressionNew}) and we developed an algorithm to calculate the deterministic schedule needed to achieve this approximate capacity in $O(N^2)$ time.

A line of work in routing algorithms~\cite{AODV,OLSR,DSR} focused on finding a route between the source and the destination while assuming that the point-to-point link capacities can only take values of 0 or 1. 
Under this assumption, the route selection based on the FD capacities is optimal. Finding the route with the largest FD capacity is equivalent to the problem of finding the widest-path between a pair of vertices in a graph~\cite{pollack1960letter}. This can be solved efficiently (i.e., in polynomial-time) by adapting any algorithm that finds the shortest path between a pair of vertices in a graph (e.g., Dijkstra's algorithm~\cite{dijkstra1959note}).
For routing with multi-rate link capacities, several heuristic metrics have been proposed to enhance the selection of routes in an ad-hoc wireless network~\cite{awerbuch2004high,de2005high}. 

Differently, we are here interested in selecting the route with the largest HD approximate 
capacity, by also trying to
fundamentally address the complexity of finding such a path.

\smallskip
\noindent{\bf{Paper Organization.}} 
Section~\ref{sec:model} describes the setting of our problem and the known capacity results for HD line networks.
Section~\ref{sec:mainresults} proves the NP-hardness of finding the best HD route in a network.
Section~\ref{sec:poly_agorithm} describes special network classes for which a polynomial-time algorithm for finding the best HD route exists.
Section~\ref{sec:conc} concludes the paper.

\section{System Model}
\label{sec:model}
We consider a network represented by the directed graph $\mcal{G}$ where $\mcal{V}(\mcal{G})$ and $\mcal{E}(\mcal{G})$ are the set of vertices (communication nodes) and the set of edges (point-to-point links) in $\mcal{G}$, respectively. 
The point-to-point link between any two nodes is assumed to be a discrete memoryless channel.
For each edge $e \in \mcal{E}(\mcal{G})$, we represent the point-to-point link capacity with $c(e) > 0$.
Over this graph, information flows from a source node $S \in \mathcal{V}(\mathcal{G})$ to a destination node $D \in \mathcal{V}(\mathcal{G})$.
For the graph $\mathcal{G}$ with $N+2$ vertices, we denote the source vertex as $v_0$ and the destination vertex as $v_{N+1}$. 

A path $\mathcal{P} = v_{k_1}~-~v_{k_2}~-~\dots~-~v_{k_{m+1}}$ of length $m$ in $\mcal{G}$ is a sequence of  vertices $v_{k_i} \in \mathcal{V}(\mathcal{G}), \forall i \in [1:m+1]$.
An $S\mbox{-}D$ {\em simple path} in $\mathcal{G}$ is a path for which $v_{k_1} = v_0$ and $v_{k_m+1} = v_{N+1}$ {and all $m+1$ vertices in $\mathcal{P}$ are distinct}, i.e., there are no cycles in $\mathcal{P}$.
The HD approximate capacity of the $S\mbox{-}D$ simple path $\mcal{P}$ is~\cite{ARXIV_version} 
\begin{align}
    \mathsf{C}_{\mathcal{P}} = \min_{i\in[2:m]}\left\{\frac{ c(e_{k_{i-1},k_{i}})\ c(e_{k_{i},k_{i+1}})}{ c(e_{k_{i-1},k_{i}})+ c(e_{k_{i},k_{i+1}}) }  \right\},
            \label{eq:HD_capacity}
        \end{align}
where $e_{k_{i-1},k_{i}}, i \in [1:m+1],$ represents the edge from node $v_{k_{i-1}} \in \mathcal{P}$ to node $v_{k_i} \in \mathcal{P}$.
The approximate capacity expression in~\eqref{eq:HD_capacity} is half of the minimum harmonic mean of the capacities of each two consecutive edges in $\mathcal{P}$.

\begin{rem}
{\rm 
The expression in~\eqref{eq:HD_capacity} was derived in~\cite{ARXIV_version} in the context of Gaussian noise networks. However, the analysis directly extends if we replace each of the Gaussian point-to-point link with any discrete memoryless channel. 
This holds since the cut-set bound is tight for any HD line network (i.e., degraded relay network~\cite{aref1981information}) where point-to-point links are discrete memoryless channels assuming we allow for dynamic schedules~\cite{KramerAllerton2004}.
As a result, deterministic schedules would only reduce the capacity by at most 1 bit per relay, thus providing the same constant-gap approximation needed in~\cite{ARXIV_version}. 
} 
\end{rem}


    \section{HD Routing is NP-hard}
\label{sec:mainresults}
IN this section, our goal in this section is to prove that the problem of finding the best HD route in a network is NP-hard. 
Towards this end, we start by showing that, if we want to find the path $\mathcal{P}$ with the largest value of $\mathsf{C}_{\mathcal{P}}$ in~\eqref{eq:HD_capacity}, then it is necessary to restrict our search over {\it simple} paths.

\subsection{Non-simple paths are misleading in HD}\label{sec:misleading_cyclic}
    Practically, a communication route through a network is expected to be a simple path, i.e, a path that contains no cycles.
    This is due to the fact that for a non-simple path, e.g.,  $\mcal{P}_{\rm cyclic} = S-v_{1}-v_2-\dots-v_m-v_2-D$, we know that -- from the degraded nature of the network -- the information sent from $v_m$ to $v_2$ is a noisy version of the information that is already available at $v_2$ (since $v_2$ appeared earlier in the path).
    Thus, for the simple path $\mcal{P}_{\rm simple} = S-v_1-v_2-D$, we fundamentally have that 
    \begin{align}
        \msf{C}_{\mcal{P}_{\rm cyclic}} \leq \msf{C}_{\mcal{P}_{\rm simple}}.
        \label{eq:simple_vs_cyclic}
    \end{align}
     This observation is true for both FD and HD paths in the network and therefore the best path (in FD or HD) is naturally a simple path.
     When routing using the FD capacities (to select the best FD route), this observation turns out to be just a technicality since the expression for the FD capacity already exhibits the fundamental property described in~\eqref{eq:simple_vs_cyclic}. 
         Particularly, we have that $\mcal{E}(\mcal{P}_{\rm simple}) \subseteq \mcal{E}(\mcal{P}_{\rm cyclic})$, which directly implies that
         \begin{align*}
             \mathsf{C}^{\rm FD}_{\mcal{P}_{\rm cyclic}} {=} \min_{e \in \mcal{E}(\mcal{P_{\rm cyclic}})}\left\{ c(e)\right\} \leq 
             \min_{e \in \mcal{E}(\mcal{P_{\rm simple}})}\left\{ c(e)\right\} = \mathsf{C}^{\rm FD}_{\mcal{P}_{\rm simple}}.
         \end{align*}
         Thus, an algorithm that selects a route in FD can end up with {either} type of paths (simple or cyclic). If the path is cyclic, then we can prune it to get a simple path while ensuring that pruning can only improve the computed capacity.

   \begin{figure}[t]
        \centering
        \includegraphics[width=0.5\textwidth]{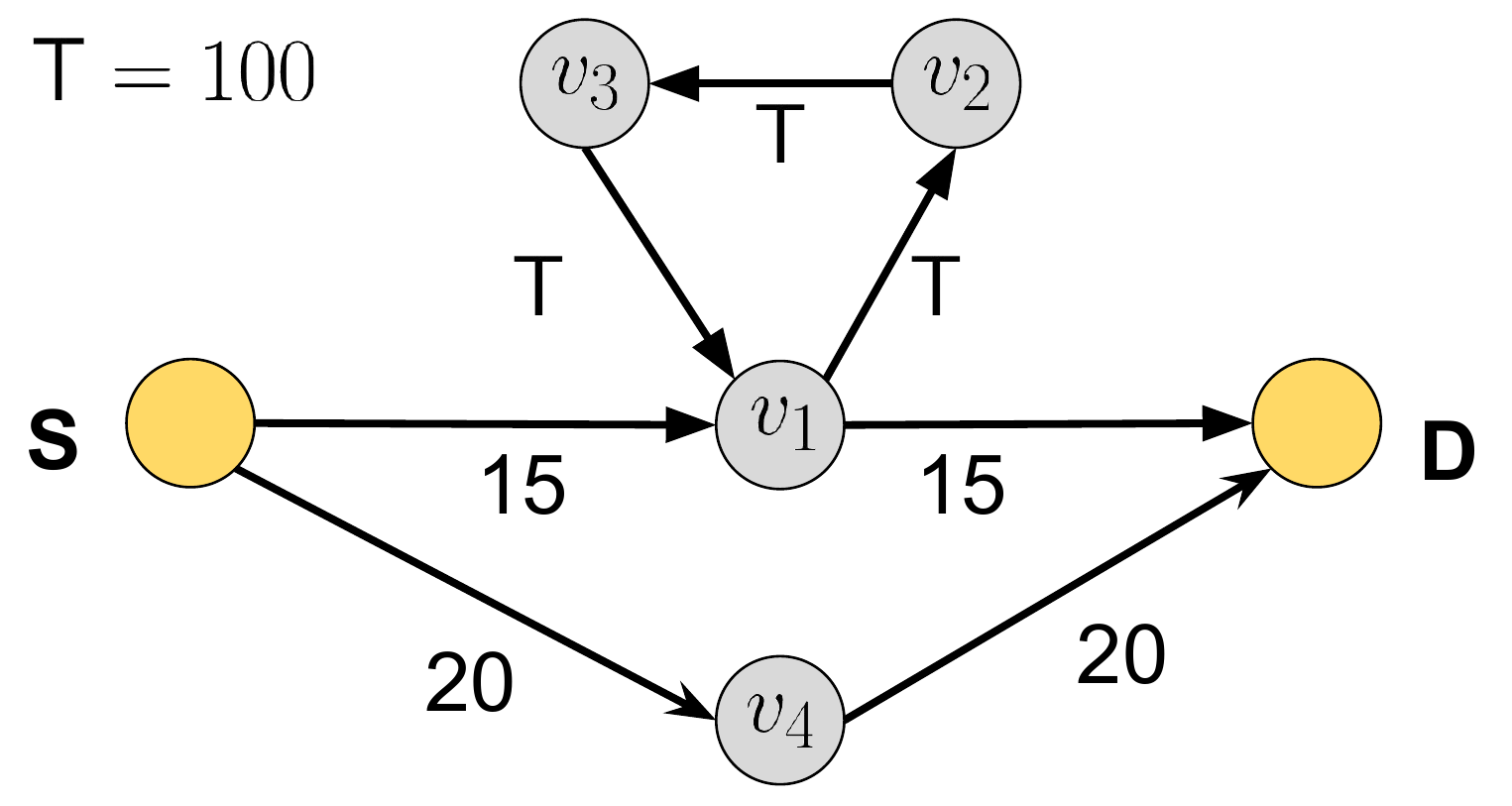}
        \caption{A network example in which a non-simple path can appear to have a larger HD approximate capacity than its simple subpath.}
        \label{fig:cycles}
    \end{figure}

Differently, for HD routing, it is very important to restrict ourselves to searching over simple paths as the HD approximate capacity expression in~\eqref{eq:HD_capacity} only applies to simple paths. Furthermore, applying the expression in~\eqref{eq:HD_capacity} to a path with a cycle can actually increase the approximate capacity (in contradiction to the fundamental property in~\eqref{eq:simple_vs_cyclic}).
     To illustrate this, consider the network example shown in Fig.~\ref{fig:cycles}.
     From Fig.~\ref{fig:cycles}, we now focus on the two paths: the simple path $\mcal{P}_1 = S - v_{1} - D$ and the non-simple path $\mcal{P}_2 = S - v_1 - v_2 - v_3 - v_1 - D$.
     Note that $\mcal{P}_1$ is a simple path and $\mcal{P}_2$ is a cyclic extension of $\mcal{P}_1$ by adding the cycle $v_1 - v_2 - v_3 - v_1$. 
     If we apply the expression in~\eqref{eq:HD_capacity} on both paths, we get the value equal to $7.5$  for $\mcal{P}_1$ and for $\mcal{P}_2$ we get $13.05$.
     Thus, if an algorithm is allowed to consider non-simple paths, then it would output the path $\mcal{P}_2$ even though we know fundamentally that $\mathsf{C}_{\mcal{P}_1} \geq \mathsf{C}_{\mcal{P}_2}$. 
     This is the first major problem that arises when we allow an algorithm to output non-simple paths based on the expression in~\eqref{eq:HD_capacity}.
     The second problem arises when we observe that $\mcal{P}_3 = S - v_4 - D$ in Fig.~\ref{fig:cycles} is actually the best HD simple path from $S$ to $D$.
     However, since applying~\eqref{eq:HD_capacity} for $\mcal{P}_2$ yields 13.05, which is larger than what we get for $\mcal{P}_3$ (i.e., $10$), then the algorithm will output a non-simple path $\mcal{P}_2$ which when pruned does not yield the best HD path. Thus, an algorithm designed with the goal to find the best HD path needs to be aware of the type of paths that it processes. 
In other words, we can no longer rely on pruning non-simple paths that an algorithm outputs as these non-simple paths in HD can mislead the algorithm into not selecting the best HD path as illustrated in this example.

     As a consequence of the above discussion, in the rest of the section, we focus on the problem of finding the simple (i.e., acyclic) path with the largest HD approximate capacity.


     \subsection{Find the best HD simple path is NP-hard}
Our goal in this subsection, is to prove that the \emph{search} problem of finding the $S\mbox{-}D$ simple path with the largest HD approximate capacity in a graph is NP-hard. Towards proving this, we first show that the related decision problem ``HD-Path'', which is defined below, is NP-complete. 

   \begin{figure*}[t]
        \centering
        \includegraphics[width=\textwidth]{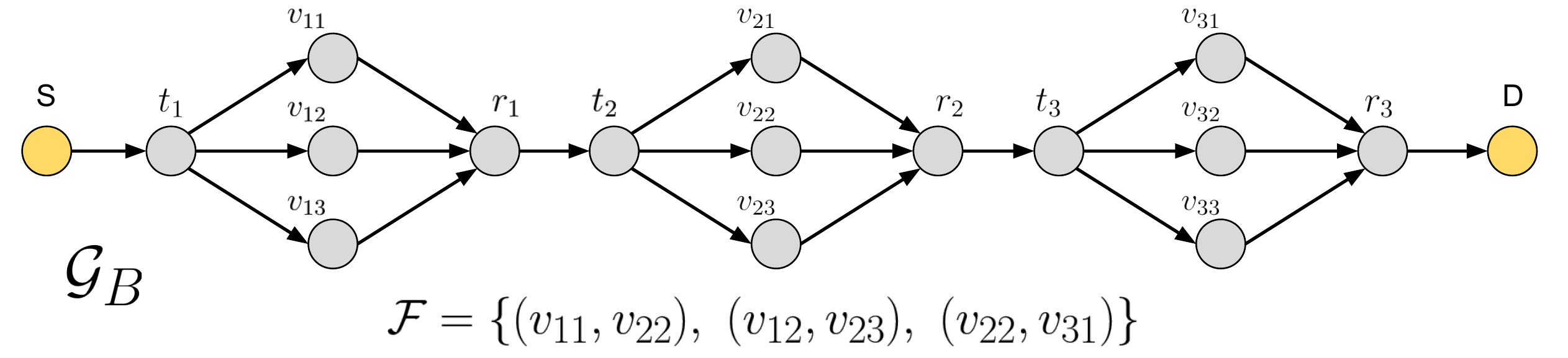}
        \caption{Graph ${\mathcal{G}}_B$ and set of forbidden pairs $\mathcal{F}$ for the boolean expression in~\eqref{eq:boolexp}.
}
        \label{fig:graph_3SAT}
    \end{figure*}

\begin{defin}[HD-Path problem]
\label{def:decPr}
Given a directed graph $\mathcal{G}$ and a scalar $\mathsf{Z}>0$, does there exist an $S\mbox{-}D$ simple path in $\mcal{G}$ with an HD approx. capacity greater than or equal to $\mathsf{Z}$?
\end{defin} 
Since the decision problem defined above can be reduced in polynomial-time to finding the $S-D$ simple path with the largest HD approximate capacity, then by proving the NP-completeness of the decision problem, we also prove that the search problem is NP-hard.

The HD-Path problem is NP because, given a guess for a path, we can verify in polynomial-time whether it is simple (i.e., no repeated vertices) and whether its HD approximate capacity is greater than or equal to $\mathsf{Z}$ by simply evaluating the expression in~\eqref{eq:HD_capacity}.

To prove the NP-completeness of the HD-Path problem, we now show that the classical 3SAT problem (which is NP-complete)~\cite{karp1972reducibility} can be reduced in polynomial-time to the HD-Path decision problem in Definition~\ref{def:decPr}.
For {\rm 3SAT}, we are given a boolean expression $B$ in 3-conjunctive normal form,
\begin{align}
    B(x_1,x_2,\dots,x_n) =& (p_{11} \vee p_{12} \vee p_{13}) \wedge (p_{21} \vee p_{22} \vee p_{23})\nonumber \\ 
    &\wedge \dots \wedge (p_{m1} \vee p_{m2} \vee p_{m3}),
    \label{eq:3SAT_expr}
\end{align}
where: (i) $B$ is a conjunction of $m$ clauses $\{C_1,C_2,\dots,C_m\}$, each containing
a disjunction of three literals
and (ii) a literal $p_{ij}$ is either a boolean variable $x_k$ or its negation $\bar{x}_k$ for some $k \in [1:n]$.
The boolean expression $B$ is \emph{satisfiable} if the variables $x_{[1:n]}$ can be assigned boolean values so that $B$ is true.
The 3SAT problem answers the question: \emph{Is the given $B$ satisfiable?}
We next prove the main result of this section through the following lemma.
\begin{lem}
\label{lemm:lemmMain}
    There exists a polynomial-time reduction from the {\rm 3SAT} problem to the {\rm HD-Path} problem.
\end{lem}
\begin{proof}
    To prove this statement, we are going to create a sequence of graphs based on the boolean statement $B$ given to the 3SAT problem.
    In each of these graphs, we show that the existence of a satisfying assignment for $B$ is equivalent to a particular feature in the graph.
Finally, we construct
an HD network where the feature equivalent to a satisfying assignment of $B$ is to find {a simple path with HD approximate} capacity greater than or equal to $\msf{Z}$.
In particular, our proof follows four steps of graph constructions, which are explained in details in what follows.
To illustrate these four steps we use the following boolean expression as a running example:
\begin{align}
\label{eq:boolexp}
B = (\bar{x}_{1} \vee x_2 \vee x_3)\wedge(x_4 \vee x_1 \vee \bar{x}_2) \wedge (\bar{x}_1 \vee x_3 \vee \bar{x}_5),
\end{align}
where, with the notation in~\eqref{eq:3SAT_expr}, the literals are assigned as
\begin{subequations}
\label{eq:exlitas}
\begin{align}
&(p_{11},p_{12},p_{13}) = (\bar{x}_1, x_2, x_3), 
\\&(p_{21},p_{22},p_{23}) = (x_4,x_1,\bar{x}_2),
\\&(p_{31},p_{32},p_{33}) = (\bar{x}_1,x_3,\bar{x}_5).
\end{align}
\end{subequations}

\noindent {\bf 1)}  Assume that the boolean expression $B$ is made of $m$ clauses. 
    For each clause $C_i, i \in [1:m],$ in $B$, construct a gadget digraph $\mathcal{G}_i$ with vertices $\mathcal{V}(\mathcal{G}_i) =\{t_i,v_{i1},v_{i2},v_{i3},r_i\}$ and edges $\mathcal{E}(\mathcal{G}_i) = \bigcup_{j=1}^3 \left \{ e_{t_i ,v_{ij}}, e_{v_{ij}, r_i} \right \}$.
    Now we connect the gadget graphs $\mathcal{G}_i, i \in [1:m],$ by adding directed edges $e_{r_i, t_{i+1}}$,
$\forall i \in [1:m{-}1]$. 
    Finally, we introduce a source vertex $S$ and a destination vertex $D$ and the directed edges $e_{S, t_1}$ and $e_{r_m,D}$.
    We denote this new graph construction by $\mathcal{G}_B$.
    Note that each vertex $v_{ij}$ in $\mathcal{G}_B$ represents a literal $p_{ij}$ in the boolean expression $B$. 
We call a pair of vertices $(v_{ij},v_{k\ell})$ in $\mathcal{G}_B$, with $i < k$, as \emph{forbidden} if $p_{ij} = \widebar{p_{k\ell}}$ in $B$.

    Let $\mcal{F}$ be the set of all such forbidden pairs. 
Consider an $S\mbox{-}D$ path $\mathcal{P} = S - t_1 - v_{1\ell_1} - r_1 - t_2 - \dots - v_{m\ell_m} - r_m - D$ in the graph $\mathcal{G}_B$ that contains at most one vertex from any forbidden pair in $\mcal{F}$. 
    Using the indexes characterizing the path $\mathcal{P}$, if we set the literals $p_{i\ell_i}$ to be true $\forall i \in [1:m]$, then this is a valid assignment (since, by our definition, $\mathcal{P}$ avoids all forbidden pairs in $\mcal{F}$). 
    Additionally, since we set one literal to be true in each clause, then this assignment satisfies $B$. Hence the existence of a path $\mathcal{P}$ in $\mathcal{G}_B$ that avoids forbidden pairs implies that $B$ is satisfiable.
    Similarly, we can show that if $B$ is satisfiable, then we can construct a path that avoids forbidden pairs in $\mathcal{G}_B$ using any assignment that satisfies $B$.

        \smallskip
\noindent    {\bf Example.} 
The boolean expression in~\eqref{eq:boolexp} has $m=3$ clauses. Hence, we construct $3$ gadget digraphs that are connected to form $\mathcal{G}_B$ as represented in Fig.~\ref{fig:graph_3SAT}.
Since each vertex $v_{ij}, i \in [1:m],\ j\in [1:3]$, in $\mathcal{G}_B$ represents a literal $p_{ij}$ in the boolean expression in~\eqref{eq:boolexp} (i.e., $p_{ij}=v_{ij}$) and the literals are assigned as described in~\eqref{eq:exlitas}, then the set of forbidden pairs is given by
\begin{align}
\mathcal{F} = \left \{ (v_{11},v_{22}), (v_{12},v_{23}),(v_{22},v_{31}) \right \}
\label{eq:forbpairs}
\end{align}
as also shown in Fig.~\ref{fig:graph_3SAT}.

%

\medskip

\begin{figure*}[t]
        \centering
        \includegraphics[width=\textwidth]{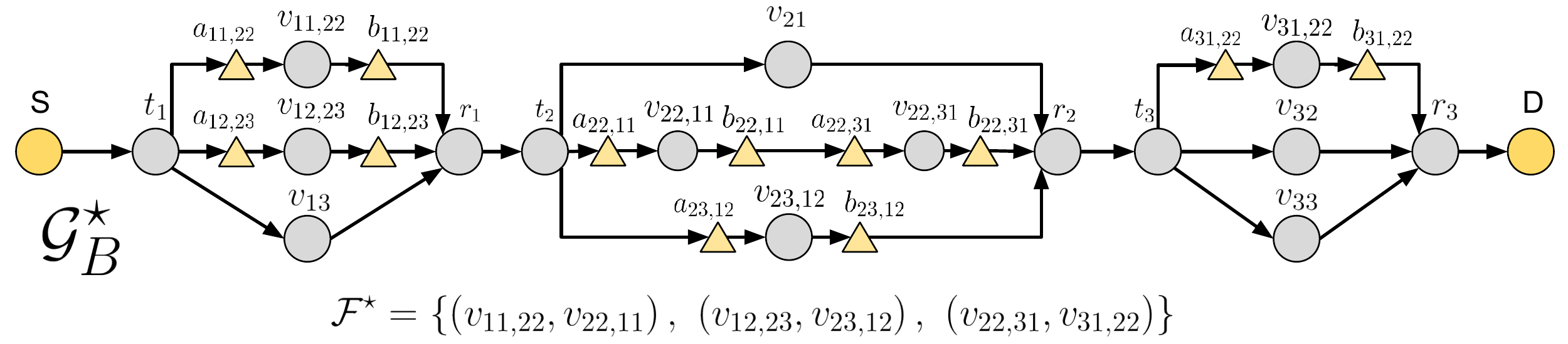}
        \caption{Graph $\mathcal{G}^\star_B$ and set of forbidden pairs $\mathcal{F}^\star$. The graph $\mathcal{G}^\star_B$ is constructed from $\mathcal{G}_B$ in Fig.~\ref{fig:graph_3SAT}.}
        \label{fig:graph_3SAT_modified1}
    \end{figure*}
\noindent {\bf 2)}    Next we modify the set of forbidden pairs $\mcal{F}$ and the graph $\mathcal{G}_B$ such that each vertex appears at most once in $\mcal{F}$.
    For each vertex $v_{ij}$ that appears in at least one forbidden pair of $\mcal{F}$, define $\mathcal{V}_{\mcal{F}}(v_{ij}) =\{v_{i'j'} \in \mathcal{V} (\mathcal{G}_{B}) | (v_{ij},v_{i'j'}) \in \mcal{F} \}$.
Then, for each $\mathcal{V}_{\mcal{F}}(v_{ij})$, we create $\left | \mathcal{V}_{\mcal{F}}(v_{ij}) \right |$ vertices and we label them as $v_{ij,k\ell}$, $\forall v_{k\ell} \in \mathcal{V}_{\mcal{F}}(v_{ij})$.
We finally replace the vertex $v_{ij}$ in $\mathcal{G}_B$ with a path connecting the vertices $v_{ij,k\ell}$, $\forall v_{k\ell} \in \mathcal{V}_{\mcal{F}}(v_{ij})$. 
    We denote this new graph as $\mathcal{G}_B^{\circ}$.
The new set of forbidden pairs $\mcal{F}^\circ$ is defined based on the set $\mcal{F}$ as $\mcal{F}^\circ = \left\{\left. \left(v_{ij,k\ell}, v_{k\ell,ij}\right) \right|\ (v_{ij},v_{k\ell}) \in \mcal{F}  \right\}$.
    Note that, for this new set of forbidden pairs, each vertex in $\mathcal{G}^\circ_B$ appears in at most one forbidden pair.
    Let $\mathcal{V}_{\mcal{F}^\circ}$ be the set of vertices that appear in $\mcal{F}^\circ$.
    Then $\forall v_{ij,kl} \in \mathcal{V}_{\mcal{F}^\circ}$, we replace $v_{ij,kl}$ with a path that consists of three vertices. 
    In particular, for any vertex $v_{ij,k\ell} \in \mathcal{V}_{\mcal{F}^\circ}$, we replace it with a directed path $a_{ij,kl} -  v_{ij,kl} - b_{ij,kl}$. 
    We call this new graph $\mathcal{G}_B^{\star}$ and the forbidden pair set $\mcal{F}^\star = \mcal{F}^\circ$.
    The newly introduced vertices $a_{ij,k\ell}$ and $b_{ij,k\ell}$ are called \emph{a-type} and \emph{b-type} vertices, respectively.
    
Similar to our earlier argument for $\mathcal{G}_B$, note that a path in $\mathcal{G}_B^\star$ that avoids forbidden pairs in $\mcal{F}^\star$ gives a valid satisfying assignment for the boolean argument $B$.
In the reverse direction, if we have an assignment that satisfies $B$, then by taking one true literal from each clause $C_i$, $i\in [1:m]$, we can choose $t_i-r_i$ paths that avoid forbidden pairs. By connecting these paths together, we get an $S\mbox{-}D$ path in $\mathcal{G}_B^\star$ that avoids forbidden pairs.

\smallskip

\noindent {\bf Example.} For our running example, given the set of forbidden pairs $\mathcal{F}$ in~\eqref{eq:forbpairs}, we have
\begin{align*}
&\mathcal{V}_{\mathcal{F}} (v_{11}) = \left \{ v_{22} \right \} \implies v_{11} \leftarrow v_{11,22},
\\
&\mathcal{V}_{\mathcal{F}} (v_{22}) = \left \{ v_{11},v_{31}\right \} \implies v_{22} \leftarrow v_{22,11} - v_{22,31},
\\
&\mathcal{V}_{\mathcal{F}} (v_{12}) = \left \{ v_{23}\right \} \implies v_{12} \leftarrow v_{12,23},
\\
&\mathcal{V}_{\mathcal{F}} (v_{23}) = \left \{ v_{12}\right \} \implies v_{23} \leftarrow v_{23,12},
\\
&\mathcal{V}_{\mathcal{F}} (v_{31}) = \left \{ v_{22}\right \} \implies v_{31} \leftarrow v_{31,22},
\end{align*}
where $y \leftarrow \mathcal{Y}$ indicates that in $\mathcal{G}_B^{\circ}$ the vertex $y$ is replaced by the path $\mathcal{Y}$.
The set of forbidden pairs $\mcal{F}^\circ$ is then given by
\begin{align}
\mcal{F}^\circ = \left \{ (v_{11,22},v_{22,11}),(v_{22,31},v_{31,22}),(v_{12,23},v_{23,12}) \right \}
\label{eq:Fcirc}
\end{align}
and hence $\mathcal{V}_{\mcal{F}^\circ} = \left \{  v_{11,22},v_{22,11},v_{22,31},v_{31,22},v_{12,23},v_{23,12} \right \}$.
Given this, we can now construct the graph $\mathcal{G}_B^{\star}$ by replacing any vertex inside $\mathcal{V}_{\mcal{F}^\circ}$ as follows
\begin{align*}
&v_{11,22} \leftarrow a_{11,22} -  v_{11,22} - b_{11,22},
\\
&v_{22,11} \leftarrow a_{22,11} -  v_{22,11} - b_{22,11},
\\
&v_{22,31} \leftarrow a_{22,31} -  v_{22,31} - b_{22,31},
\\
&v_{31,22} \leftarrow a_{31,22} -  v_{31,22} - b_{31,22},
\\
&v_{12,23} \leftarrow a_{12,23} -  v_{12,23} - b_{12,23},
\\
&v_{23,12} \leftarrow a_{23,12} -  v_{23,12} - b_{23,12},
\end{align*}
as shown in Fig.~\ref{fig:graph_3SAT_modified1}. Furthermore, we have $\mcal{F}^\circ = \mcal{F}^\star$, where $\mcal{F}^\circ$ is defined in~\eqref{eq:Fcirc}.

\medskip

\noindent {\bf 3)} Our next step is to modify $\mathcal{G}_B^\star$ to incorporate $\mcal{F}^\star$ directly into the structure of the graph. 
For each $(v_{ij,k\ell},v_{k\ell,ij}) \in \mcal{F}^\star$ introduce a new vertex $f_{ij,k\ell}$ to replace $v_{ij,k\ell}$ and $v_{k\ell,ij}$. 
All edges that were incident from (to) $v_{ij,k\ell}$ and $v_{k\ell,ij}$ are now incident from (to) $f_{ij,k\ell}$.
We call these newly introduced vertices as \emph{f-type} vertices and denote this new graph as $\mathcal{G}^\bullet_B$.
Note that in $\mathcal{G}^\bullet_B$, we now have incident edges from $a_{ij,k\ell}$ and $a_{k\ell,ij}$ to $f_{ij,k\ell}$ and edges incident from $f_{ij,k\ell}$ to vertices $b_{ij,k\ell}$ and $b_{k\ell,ij}$.
A path in $\mathcal{G}^\star_B$ that avoids forbidden pairs in $\mcal{F}^\star$ gives a path in $\mathcal{G}^\bullet_B$ that follows the following rules:
\begin{enumerate}
    \item {\bf Rule 1:} If any \emph{f-type} vertex is visited, then it is visited at most once;
    \item {\bf Rule 2:} If an \emph{f-type} vertex is visited then the preceding \emph{a-type} vertex and the following \emph{b-type} vertex both share the same index (i.e., we do not have $a_{ij,k\ell} - f_{ij,k\ell} - b_{k\ell,ij}$ or $a_{k\ell,ij} - f_{ij,k\ell} - b_{ij,k\ell}$ as a subpath of our path in $\mathcal{G}_B^\bullet$).
\end{enumerate}

It is not difficult to see that an $S\mbox{-}D$ path in $\mathcal{G}_B^\bullet$ that abides to the two aforementioned rules represents a feasible path that avoids forbidden pairs $\mcal{F}^\star$ in $\mathcal{G}_B^\star$. Specifically,
this can be seen by treating the subpath ($a_{ij,k\ell} - f_{ij,k\ell} - b_{ij,k\ell}$) in $\mathcal{G}_B^\bullet$ as the subpath ($a_{ij,k\ell} - v_{ij,k\ell} - b_{ij,k\ell}$) in $\mathcal{G}^\star_B$ and similarly ($a_{k\ell,ij} - f_{ij,k\ell} - b_{k\ell,ij}$) for ($a_{k\ell,ij} - v_{k\ell,ij} - b_{k\ell,ij}$).
In other words, the problem of finding a path in $\mathcal{G}_B^\star$ that avoids forbidden pairs in $\mcal{F}^\star$ is equivalent to finding a path in $\mathcal{G}_B^\bullet$ that satisfies Rule 1 and Rule 2.

        \smallskip
        \noindent {\bf Example.} For our running example, the graph $\mathcal{G}_B^\bullet$ is shown in Fig.~\ref{fig:graph_3SAT_modified2}. 
In particular, $\mathcal{G}_B^\bullet$ is constructed from $\mathcal{G}^\star_B$ in Fig.~\ref{fig:graph_3SAT_modified1}, where each $v_{ij,k\ell} \in \mathcal{V}(\mathcal{G}_B^\star)$ and $v_{k\ell,ij}  \in \mathcal{V}(\mathcal{G}_B^\star)$ such that $(v_{ij,k\ell},v_{k\ell,ij}) \in \mcal{F}^\star$, with $\mcal{F}^\star$ being defined in~\eqref{eq:Fcirc}, is now replaced by $f_{ij,k\ell}$ in $\mathcal{G}_B^\bullet$, which is connected to the other nodes as explained above.

\medskip

\noindent {\bf 4)} Our next step is to modify $\mathcal{G}^\bullet_B$ by introducing edge capacities.
For any edge $e \in \mathcal{E}(\mathcal{G}_B^\bullet)$ that is not incident from or to an \emph{f-type} vertex, we set the capacity of that edge  $c(e) = 3\msf{Z}$.
For an \emph{f-type} vertex $f_{ij,k\ell}$, let $g_1$ and $h_1$ be the edges incident to it from $a_{ij,k\ell}$ and incident from it to $b_{ij,k\ell}$, respectively. Similarly, let $g_2$ and $h_2$ be the edges incident from $a_{k\ell,ij}$ and to $b_{k\ell,ij}$, respectively.
Then, we set the edge capacities of these edges as
\begin{align*}
c(g_1) = c(h_2) = 1.5\msf{Z},\quad c(g_2) = c(h_1) = 3\msf{Z}.
\end{align*}

    \begin{figure*}[t]
        \centering
        \includegraphics[width=\textwidth]{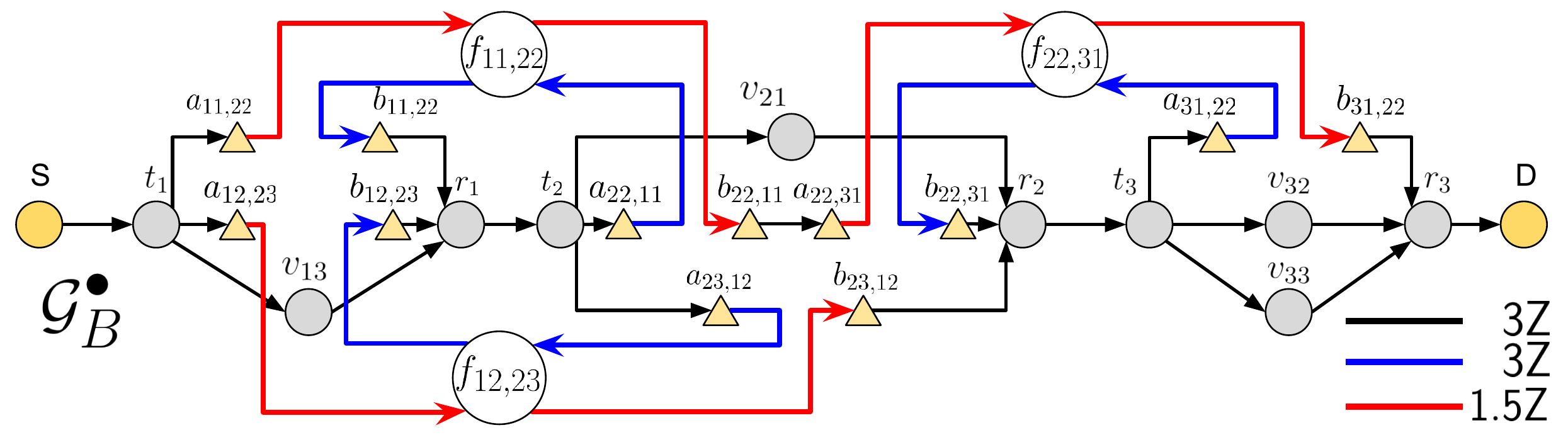}
        \caption{Graph $\mathcal{G}^\bullet_B$ and the associated edge capacities. The graph $\mathcal{G}^\bullet_B$ is constructed from $\mathcal{G}^\star_B$ in Fig.~\ref{fig:graph_3SAT_modified1}.}
        \label{fig:graph_3SAT_modified2}
    \end{figure*}
We now need to show that finding a path satisfying Rules~1 and 2 is equivalent to finding a simple path in $\mathcal{G}_B^\bullet$ with HD approximate capacity greater than or equal to $\mathsf{Z}$.
It is not difficult to see that a path that follows Rules 1 and 2 is simple and has an HD approximate capacity greater than or equal to $\mathsf{Z}$ (by avoiding subpaths $a_{ij,k\ell} - f_{ij,k\ell} - b_{k\ell,ij}$). To prove the equivalence, we now need to show that a simple path with capacity greater than or equal to $\mathsf{Z}$ satisfies Rules 1 and 2. Towards this end, note that Rule 1 is inherently satisfied since the path is simple (i.e., it visits any vertex at most once). 
For Rule 2, we next argue that both subpaths are avoided by contradiction. 

Assume that the simple path selected contains a subpath of the form $a_{ij,k\ell} - f_{ij,k\ell} - b_{k\ell,ij}$. By our construction of the edge capacities, both the edges $e_{a_{ij,k\ell},f_{ij,k\ell}}$ and $e_{f_{ij,k\ell},b_{k\ell,ij}}$
have a capacity equal to $1.5\mathsf{Z}$.
This gives us a contradiction since half of the harmonic mean between the capacities of these two consecutive edges is equal to $0.75\mathsf{Z}$.
Since the HD approximate capacity of a path is the minimum of half of the harmonic means of its consecutive edges, then the selected path cannot have an HD approximate capacity greater than or equal to $\mathsf{Z}$, which leads to a contraction.
Thus, a subpath $a_{ij,k\ell} - f_{ij,k\ell} - b_{k\ell,ij}$ is always avoided.
We now need to prove that also the path $a_{k\ell,ij} - f_{ij,k\ell} - b_{ij,k\ell}$ is always avoided.
Towards this end,
assume that the simple path selected with HD approximate capacity greater than or equal to $\mathsf{Z}$ contains (for some $i',j',k'$ and $\ell'$) a subpath of the form $a_{k'\ell',i'j'} - f_{i'j',k'\ell'} - b_{i'j',k'\ell'}$. 
Note that, as per our construction in the graph $\mathcal{G}_B^\bullet$, we have that $i' < k'$.
Let $i^\star$ be the smallest index $i'$ for which such a subpath exists in our selected path. 
Since for the subpath in question we have that $i^\star < k'$, then to reach $a_{k'\ell',i^\star j'}$ from $S$, we have already visited $r_{i^\star}$ earlier in the path. However, to move from $b_{i^\star j',k'\ell'}$ to $D$ (after the subpath in question), we need to pass through $r_{i^\star}$ once more. 
Clearly, since the path is simple, this leads to a contradiction.
Thus, a subpath $a_{k\ell,ij} - f_{ij,k\ell} - b_{ij,k\ell}$ is also always avoided.
This completes the proof that a simple path with capacity greater than or equal to $\mathsf{Z}$ satisfies Rule~2. 
Therefore,
finding a path satisfying Rules 1 and 2 is equivalent to finding a simple path in $\mathcal{G}_B^\bullet$ with HD approximate capacity greater than or equal to $\mathsf{Z}$. The second statement is an instance of the HD-Path problem in Definition~\ref{def:decPr}.

Note that in each of the four graph constructions described earlier, we construct one graph from the other using a polynomial number of operations. 
Thus, this proves by construction that there exists a polynomial reduction from the 3SAT problem to the HD-Path problem.
This concludes the proof of Lemma~\ref{lemm:lemmMain}.

\smallskip
\noindent {\bf Example.} For our running example the assignment of the edge capacities is shown in Fig.~\ref{fig:graph_3SAT_modified2}, where {\it black} and {\it blue} edges have a capacity of $3\mathsf{Z}$ and {\it red} edges have a capacity of $1.5 \mathsf{Z}$.
\end{proof}

\section{Some instances with polynomial-time solutions}
\label{sec:poly_agorithm}
In this section, we discuss a special class of networks for which there exists a polynomial-time algorithm to find a simple path with the largest HD approximate capacity.
In particular, we focus on networks where the number of cycles is polynomial, i.e., the number of cycles is at most $N^\alpha$ for some constant $\alpha > 0$, where $N+2$ is the total number of nodes in the network.
Our approach is based on relating paths in a network (described by the digraph $\mcal{G}$) to paths in the line digraph of $\mcal{G}$ denoted as $\mcal{L}_\mcal{G}$.
We describe the relation in the following {subsection} and then present an algorithm that finds the best HD simple path in 
polynomial-time
for the aforementioned class of networks.

\subsection{The line digraph perspective to the best HD path problem}\label{sec:chordal}
A line digraph of a digraph $\mcal{G}$ is defined as follows.
\begin{defin}[Line digraph $\mcal{L}_\mcal{G}$]\label{defin:line_graph}
    For a given digraph $\mcal{G}$, its line digraph $\mcal{L}_\mcal{G}$ is a digraph defined by the set of vertices $\mcal{V}(\mcal{L}_{\mcal{G}})$ and the set of directed edges $\mcal{E}(\mcal{L}_\mcal{G})$. 
    The set $\mcal{V}(\mcal{L}_\mcal{G})$ is defined as $\mcal{V}(\mcal{L}_\mcal{G}) = \{v_{ij} | e_{i,j} \in \mcal{E}(\mcal{G})\}$ where $e_{i,j}$ is the directed edge from vertex $v_i$ to vertex $v_j$.
    The set of edges $\mcal{E}(\mcal{L}_\mcal{G})$ is defined as $\mcal{E}(\mcal{L}_\mcal{G}) = \{e_{ij,k\ell} |k = j,\ v_{ij}, v_{k\ell} \in \mcal{V}(\mcal{L}_\mcal{G}) \}$.
\end{defin}
An illustration of a digraph and its associated line digraph is shown in Fig.~\ref{fig:example_line_graph}.
We can make the two following observations on how simple HD paths are represented in the line digraph.

\medskip

\noindent {\bf 1) HD paths in $\mcal{G}$ are equivalent to FD paths in $\mcal{L}_\mcal{G}$.}
Note that a path
$\mcal{P}$ in a network $\mcal{G}$ 
can be equivalently defined as the sequence of its adjacent edges (instead of vertices),
i.e., we can equivalently write the path 
$\mcal{P} = v_{k_1} - v_{k_2} - \dots - v_{k_m}$ in 
$\mcal{G}$ as $\mcal{P} = e_{k_1,k_2} - e_{k_2,k_3} - \dots - e_{k_{m-1},k_m}$.
Given this and from the definition of the line digraph $\mcal{L}_\mcal{G}$, the path $\mcal{P}$ in $\mcal{G}$ is equivalent to the path $\mcal{P}_\mcal{L} = v_{k_1k_2} - v_{k_2k_3} \dots - v_{k_{m-1}k_m}$ in $\mcal{L}_\mcal{G}$. 
For each edge $e_{ij,jk} \in \mcal{E}(\mcal{L}_\mcal{G})$, we define the capacity for the edge $e_{ij,jk}$ as
\begin{align}
    c_\mcal{L}(e_{ij,jk}) = \frac{c(e_{i,j})c(e_{j,k})}{c(e_{i,j})+ c(e_{j,k})}, 
\label{eq:cap_L_G}
\end{align}
where $c(e_{i,j})$ is the point-to-point link capacity of the edge (link) $e_{i,j}$ in $\mcal{G}$. 
Thus, we have that the FD capacity of the path $\mcal{P}_\mcal{L}$ in $\mcal{L}_\mcal{G}$ is given by
\begin{align}\label{eq:path_HD_path_FD_line}
    \msf{C}^{\rm FD}_{\mcal{P}_\mcal{L}} &= \min_{e_{ij,jk} \in \mcal{E}(\mcal{P}_\mcal{L})} \left\{c_\mcal{L}(e_{ij,jk}) \right\} \nonumber \\
    &=\min_{e_{ij,jk} \in \mcal{E}(\mcal{P}_\mcal{L})} \left\{ \frac{c(e_{i,j})\ c(e_{j,k})}{c(e_{i,j})+ c(e_{j,k})}\right\} = \msf{C}_{\mcal{P}},
\end{align}
where $\msf{C}_{\mcal{P}}$ is defined in~\eqref{eq:HD_capacity}.
From~\eqref{eq:path_HD_path_FD_line} and our previous discussion, we can conclude that, to find the path with the largest HD approximate capacity in the network described by the digraph $\mcal{G}$, we can first find the path in $\mcal{L}_\mcal{G}$ that has the largest FD capacity (where the link capacities in $\mcal{L}_\mcal{G}$ are defined as in~\eqref{eq:cap_L_G}) and then map this path in $\mcal{L}_\mcal{G}$ into its equivalent in $\mcal{G}$.
   \begin{figure}[t]
        \centering
        \includegraphics[width=\columnwidth]{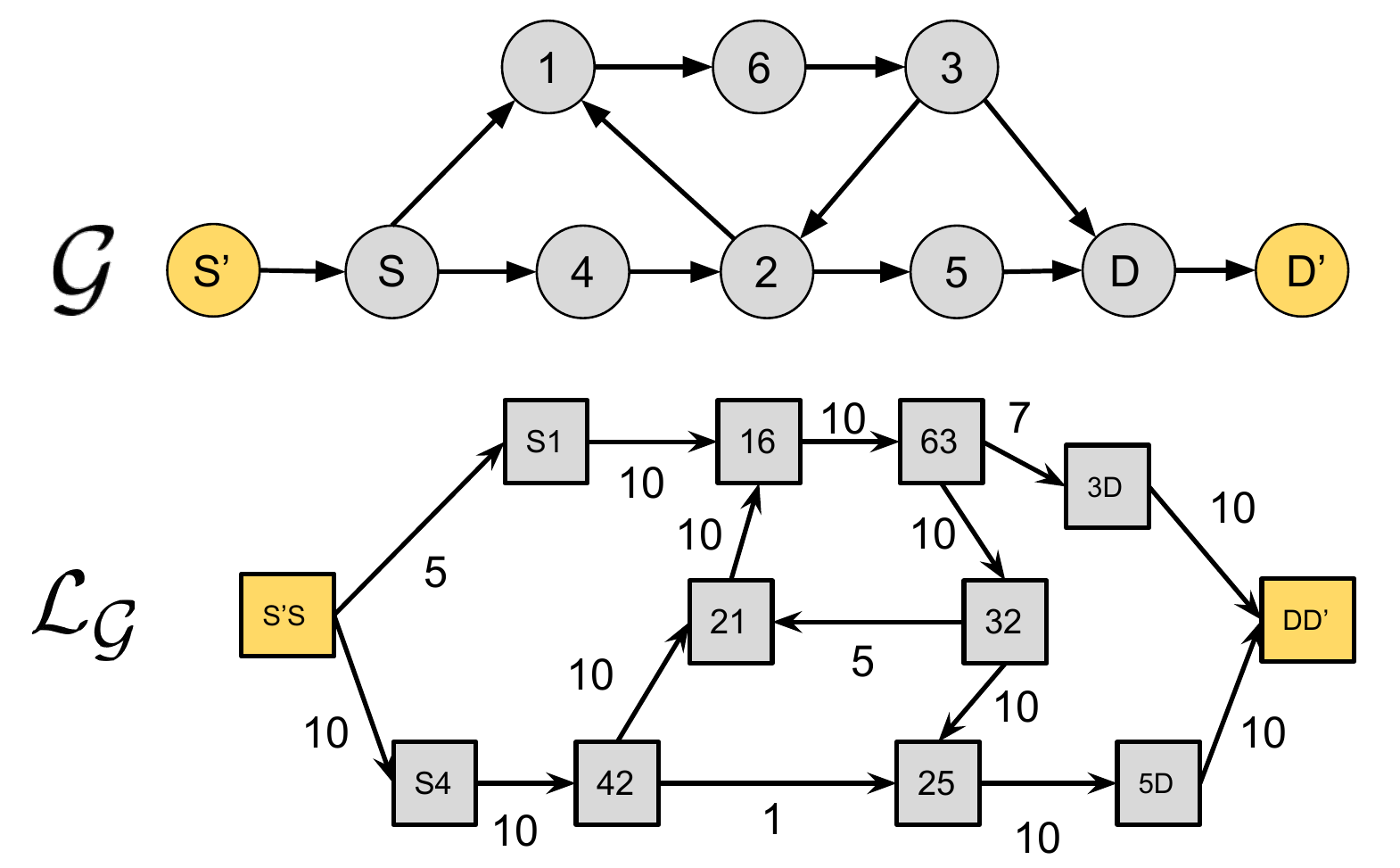}
        \caption{An example of a digraph $\mcal{G}$ with its corresponding line digraph $\mcal{L}_\mcal{G}$. For ease of notation, indexes $ij$ instead of $v_{ij}$ are used and the edge capacities are only shown on $\mcal{L}_\mcal{G}$.}
        \label{fig:example_line_graph}
    \end{figure}

\smallskip

\noindent {\bf 2) Simple paths in $\mcal{G}$ are equivalent to simple chordless paths in $\mcal{L}_\mcal{G}$.} 
We start by defining \emph{chordal} and \emph{chordless} paths in digraphs.

\begin{defin}[Chordal and chordless paths]
    A path in the digraph $\mathcal{G}^\prime$ is \emph{chordal} if there exists {an} edge $e \in \mcal{E}(\mcal{G}')$ such that its endpoints are two non-consecutive vertices in the path.
A path that is not chordal is called \emph{chordless}.
\end{defin}

For example, with reference to Fig.~\ref{fig:example_line_graph}, the path $S' - S - v_4 - v_2 - v_1  - v_6 - v_3 - D - D'$ is a chordal path in $\mcal{G}$ since $e_{3,2} \in \mcal{E}(\mcal{G})$ and the vertices $v_3$ and $v_2$ belong to the path but are non-consecutive. Thus, $e_{3,2}$ is a chord {for} this path in $\mcal{G}$.
A similar reasoning holds for $e_{S,1}$.

Consider a cyclic path $\mcal{P}_{\rm cycle}$ in $\mcal{G}$.
This implies that some vertex $v_{k} \in \mcal{P}_{\rm cycle}$ appears at least twice in the path. Denote with $v_{q_1}$ the node following $v_k$ in its first appearance in $\mcal{P}_{\rm cycle}$ and with $v_{q_2}$ the node preceding $v_k$ in its second appearance in the path $\mcal{P}_{\rm cycle}$. Then, if we write the line digraph equivalence of $\mcal{P}_{\rm cycle}$, we have
\[
    \mcal{P}_{\mcal{L}_{\rm cycle}} = \dots  - v_{k q_1} - \dots - v_{q_2 k} - \dots\ .
\]
From the construction of $\mcal{E}(\mcal{L}_\mcal{G})$ in Definition~\ref{defin:line_graph}, we see that the edge $e_{q_2k,k q_1} \in \mcal{E}(\mcal{L}_\mcal{G})$, which implies that $\mcal{P}_{\mcal{L}_{\rm cycle}}$ is chordal.

Differently, for a simple path $\mcal{P}_{\rm simple}$, any vertex $v_k \in \mcal{P}_{\rm simple}$ appears only once. Thus, in the line digraph equivalent path $\mathcal{P}_{\mcal{L}_{\rm simple}}$, the index $k$ appears only in two consecutive vertices, which implies that $\mcal{P}_{\mcal{L}_{\rm simple}}$ is chordless.
This shows the equivalence described in our observation between simple paths in $\mcal{G}$ and simple chordless paths in $\mcal{L}_\mcal{G}$.

Given the two observations above, we can now equivalently describe our HD routing problem on the line digraph as follows: \emph{Can we find the chordless simple path in $\mcal{L}_\mcal{G}$ that has the largest FD capacity?}

\subsection{An algorithm on the line digraph $\mcal{L}_\mcal{G}$}
The goal of the algorithm described in this section is to find the chordless simple path in $\mcal{L}_\mcal{G}$ that has the largest FD capacity.
The algorithm described here is a modification of the result proposed in~\cite{ahmed2009shortest} for selecting shortest paths while avoiding forbidden subpaths.
To start, we first modify our given network (described by $\mcal{G}$) so that the source $S$ and the destination $D$ have at most degree one.
In particular, we modify the digraph $\mcal{G}$ by adding two new nodes (namely, $S'$ and $D'$) that are connected only to $S$ and $D$ with edges $e_{S',S}$ and $e_{D,D'}$ (similar to Fig.~\ref{fig:example_line_graph}). 
These two added edges have point-to-point capacities equal to $X\to\infty$. 
Denote this new digraph by $\mcal{G}'$ and create the line digraph associated with $\mcal{G}'$ and denote it by $\mcal{L}_\mcal{G}^{(0)}$. In $\mcal{L}_\mcal{G}^{(0)}$, we now consider the node $v_{S'S}$ as our source and the node $v_{DD'}$ as our intended destination.


The algorithm is based on incrementally applying Dijkstra's algorithm~\cite{dijkstra1959note}.
We first try to find the best FD path from $v_{S'S}$ to $v_{DD'}$ in $\mcal{L}_\mcal{G}^{(i)}$ by running Dijkstra's algorithm. 
Note that Dijkstra's algorithm returns a spanning tree rooted at $v_{S'S}$ that describes the best FD path from $v_{S'S}$ to each vertex $v'$ in $\mcal{L}_\mcal{G}^{(i)}$. 
We denote the tree from our initial run as $\mathcal{T}_0$.
From this point, the algorithm iterates (until termination) over four main steps described as follows (starting with $i=0$).

   \begin{figure}[t]
        \centering
        \includegraphics[width=0.7\columnwidth]{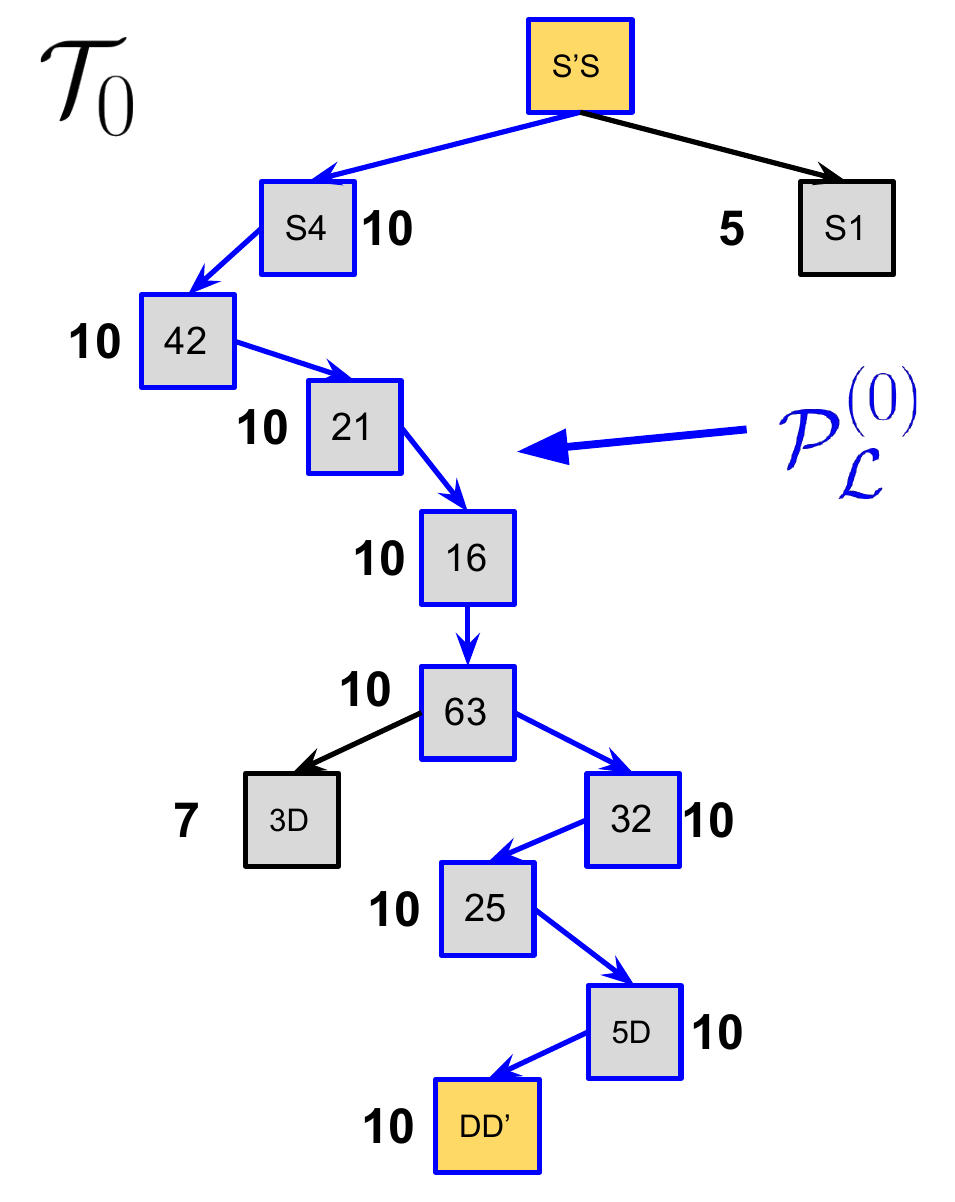}
        \caption{Spanning tree $\mcal{T}_0$ for $\mcal{L}_\mcal{G}^{(0)} = \mcal{L}_\mcal{G}$ in Fig.~\ref{fig:example_line_graph} (indexes $ij$ instead of $v_{ij}$ are used for ease of notation). 
        Boldface numbers represent the FD capacity with which a node can be reached from $SS'$ using the tree $\mcal{T}_0$. The highlighted path is the route selected from this tree $\mcal{T}_0$ from $S'S$ to $DD'$.}
        \label{fig:tree_example}
    \end{figure}

\smallskip

\noindent {\bf Step 1.} Given the line digraph $\mcal{L}^{(i)}_\mcal{G}$ and an existing best FD path spanning tree $\mcal{T}_i$, check whether the path $\mcal{P}^{(i)}_\mcal{L}$ from $v_{S'S}$ to $v_{DD'}$ defined by $\mcal{T}_i$ is chordless.
If it is chordless, terminate the algorithm since we have found the chordless path from $v_{S'S}$ to $v_{DD'}$ with the largest FD capacity. 
Otherwise, if it is not chordless, then proceed to Step 2.

\noindent{\bf Example.} We use the line digraph from Fig.~\ref{fig:example_line_graph} as our $\mcal{L}_\mcal{G}^{(0)}$. 
Then, for $i = 0$, we have the spanning tree $\mathcal{T}_0$ (from Dijkstra's algorithm) and the selected path $\mcal{P}_\mcal{L}^{(0)}$ as shown in Fig.~\ref{fig:tree_example}.
The path $\mcal{P}_\mcal{L}^{(0)}$ is chordal since 
$e_{32,21} \in \mcal{L}_\mcal{G}^{(0)}$ and $e_{42,25} \in \mcal{L}_\mcal{G}^{(0)}$.

\smallskip

\noindent {\bf Step 2.} Let $\mcal{C}^{(i)}_{\mcal{P}}$ be the set of edges in $\mcal{L}_\mcal{G}^{(i)}$ that are chords {for} the path $\mcal{P}_\mcal{L}^{(i)}$ from $v_{S'S}$ to $v_{DD'}$ discussed in the earlier step.
Let $\mcal{C}^{(i)}_{\mcal{P},\rm first} \in \mcal{C}^{(i)}_{\mcal{P}}$ be the first chord that appears along the path $\mcal{P}_\mcal{L}^{(i)}$. 
We denote the endpoints of $\mcal{C}^{(i)}_{\mcal{P},\rm first}$ as $v_{k_1k_2}$ and $v_{k_{m}k_{m+1}}$, where $v_{k_1k_2}$ is the vertex that among the two appears earlier in the path $\mcal{P}_\mcal{L}^{(i)}$ and where $m$ is the length of the subpath $\mcal{P}^{(i)}_{\rm to-fix}$ of $\mcal{P}_\mcal{L}^{(i)}$ connecting the two endpoints, i.e., we now have a path $\mcal{P}^{(i)}_{\rm to-fix} = v_{k_1k_2} - v_{k_2k_3} - \dots - v_{k_mk_{m+1}}$.
Notice that, with this, we have {$k_{m+1}=k_1$}.

\noindent{\bf Example.} For our running example and $i=0$, we can see from Fig.~\ref{fig:example_line_graph} and Fig.~\ref{fig:tree_example} that the set of chords for $\mcal{P}^{(0)}_\mcal{L}$ is $\mcal{C}^{(0)}_{\mcal{P}} = \{e_{32,21},e_{42,25} \}$. 
The selected chord $\mcal{C}^{(0)}_{\mcal{P},\rm first}$ is $e_{32,21}$ because its effect on the path concludes earlier than $e_{42,25}$.
Hence, we have
$\mcal{P}^{(0)}_{\rm to-fix} = v_{21} - v_{16} - v_{63} - v_{32}$, which is of length $m=4$.

\smallskip

\noindent {\bf Step 3.}
We now introduce new vertices to the graph $\mcal{L}_\mcal{G}^{(i)}$ by replicating every intermediate vertex in $\mcal{P}_{\rm to-fix}^{(i)}$. 
In particular, we introduce a replica vertex $v_{k'_ik'_{i+1}}$ for $v_{k_ik_{i+1}}$ where $i \in [2:m{-}1]$.
We connect these replicas of vertices to each other in the same way their corresponding originals are connected in $\mcal{P}_{\rm to-fix}^{(i)}$,
i.e., we include the edge $e_{k'_ik'_{i+1},k'_{i+1}k'_{i+2}} \forall i \in [2:m{-}1]$ {with the same edge capacity as $e_{k_ik_{i+1},k_{i+1}k_{i+2}}$}.

Then,
for every $v_{i'j'} \in \mcal{V}({\mcal{L}^{(i)}_\mcal{G}})\backslash\mcal{V}({\mcal{P}^{(i)}_{\rm to-fix}})$ such that $e_{i'j',k_ik_{i+1}} \in \mcal{E}(\mcal{L}^{(i)}_\mcal{G})$, we add an edge that connects $v_{i'j'}$ to the replica vertex of $v_{k_ik_{i+1}}$, i.e., we 
add the edge $e_{i' j',k'_ik'_{i+1}}$ {(with the same edge capacity as $e_{i' j',k_ik_{i+1}}$)}.
In other words, every vertex in $\mcal{L}_\mcal{G}^{(i)}$ that is not in $\mcal{P}_{\rm to-fix}^{(i)}$ and has an edge incident on an intermediate vertex $v_{k_ik_{i+1}}, i \in [2:m-1],$ of $\mcal{P}_{\rm to-fix}^{(i)}$ now has a similar (replicated) edge incident on the replica $v_{k'_ik'_{i+1}}$ of $v_{k_ik_{i+1}}$.
    Note that at this point: (i) the original vertices in $\mcal{P}^{(i)}_{\rm to-fix}$ still form a chordal path in $\mcal{L}_\mcal{G}^{(i)}$ and (ii) the replica vertices have every possible incident connection their original vertices had except connections to the two endpoint vertices of $\mcal{P}^{(i)}_{\rm to-fix}$.
    We denote the digraph at this point as $\widehat{\mcal{L}}_\mcal{G}^{(i+1)}$.

Now, our last change is to modify how the two endpoints of the path $\mcal{P}^{(i)}_{\rm to-fix}$ in $\widehat{\mcal{L}}^{(i+1)}_{\mcal{G}}$ connect to the intermediate vertices of the path and their replicas. 
We do this by adding the edge $e_{k'_{m-1}k'_{m},k_mk_{m+1}}$ that connects the last replicated vertex $v_{k'_{m-1}k'_m}$ to the endpoint $v_{k_mk_{m+1}}$of $\mcal{P}^{(i)}_{\rm to-fix}$ and by removing the edge $e_{k_{m-1}k_{m},k_mk_{m+1}}$ that connected the original last intermediate vertex to the endpoint. 
In particular, the new edge $e_{k'_{m-1}k'_{m},k_mk_{m+1}}$ has the same capacity as $e_{k_{m-1}k_{m},k_mk_{m+1}}$ that was removed.
Denote this new digraph as $\mcal{L}_\mcal{G}^{(i+1)}$.
Note that in this new digraph $\mcal{L}_\mcal{G}^{(i+1)}$, the path $\mcal{P}^{(i)}_{\rm to-fix}$ does not exist anymore,
while all the other chordless paths have stayed the same.
Thus, we have successfully eliminated a cycle (chordal path) that appeared in the digraph before by replicating vertices and deleting edges.

%

\noindent{\bf Example.} For our running example and $i=0$, recall that $\mcal{P}^{(0)}_{\rm to-fix} = v_{21} - v_{16} - v_{63} - v_{32}$. The new generated digraphs $\widehat{\mcal{L}}_\mcal{G}^{(1)}$ and $\mcal{L}_\mcal{G}^{(1)}$ are shown in Fig.~\ref{fig:line_graph_modified_step3}.

   \begin{figure}[t]
        \centering
        \includegraphics[width=0.48\textwidth]{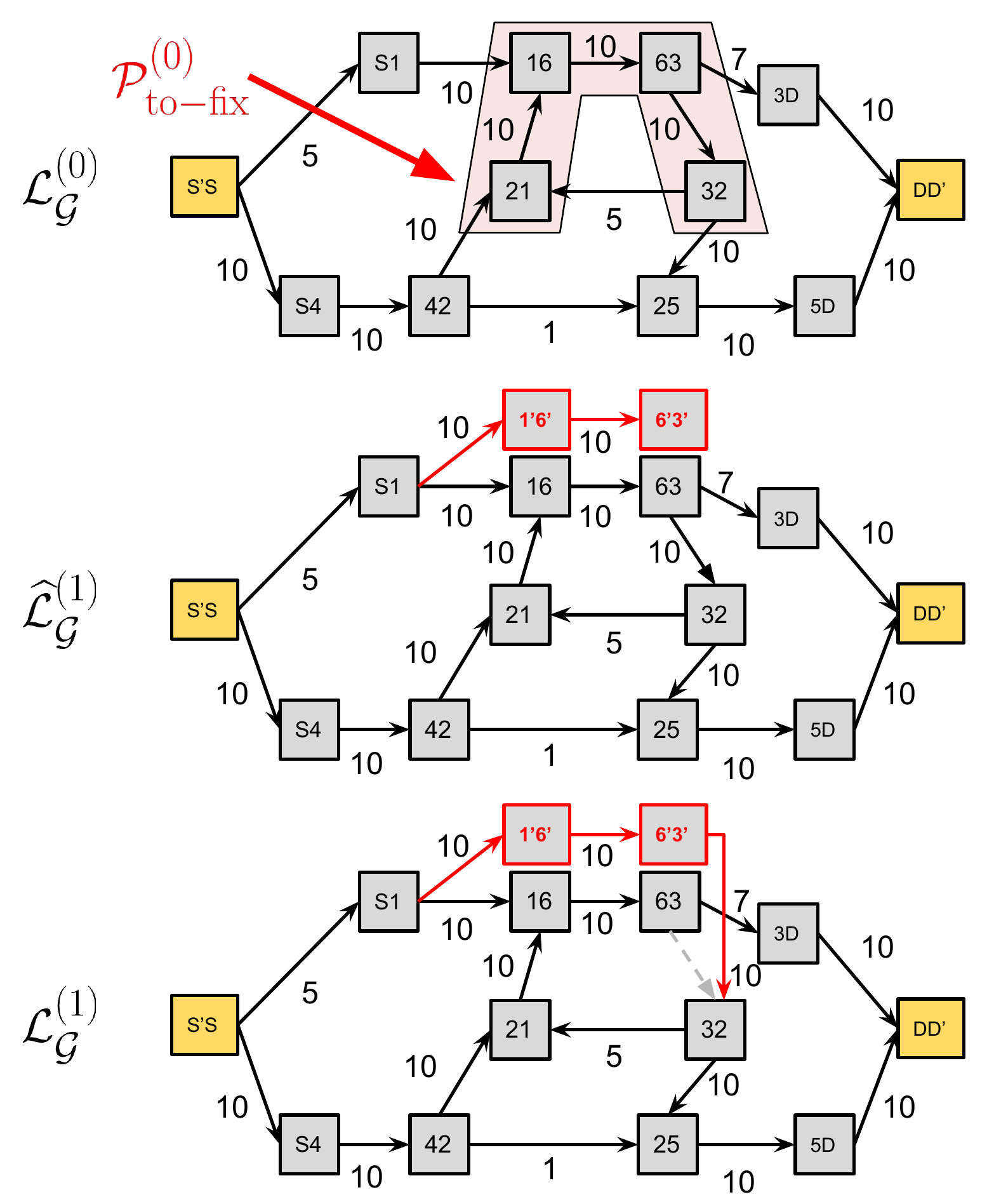}
        \caption{ $\mcal{L}_\mcal{G}^{(0)}$ from Fig.~\ref{fig:example_line_graph} and the corresponding $\widehat{\mcal{L}}_\mcal{G}^{(1)}$ and $\mcal{L}_\mcal{G}^{(1)}$. The replica vertices and the added edges are shown in red while the deleted edges are dashed.}
        \label{fig:line_graph_modified_step3}
    \end{figure}

\smallskip

\noindent {\bf Step 4.} In the fourth step, our goal is to create the spanning tree $\mcal{T}_{i+1}$ of the best FD paths associated with the digraph $\mcal{L}^{(i+1)}_\mcal{G}$.
To ensure termination of the algorithm, a condition for this construction is that $\mcal{T}_{i+1}$ should be made as similar as possible to $\mcal{T}_i$~\cite{ahmed2009shortest}. 
To do so, we 
run Dijkstra's algorithm to find the new spanning tree $\mathcal{T}_{i+1}$ but we 
start at an intermediate stage in the algorithm, since we already know part of the spanning tree from $\mcal{T}_{i}$. 
In particular, we do the following procedure. Recall our definition of $\mcal{P}^{(i)}_{\rm to-fix}$ and its endpoint $v_{k_mk_{m+1}}$ in Step 2. 
Define $\mcal{V}_{\rm redo}^{(i+1)}$ to be the set of vertices for which we need to find a new best FD path. 
In particular, we define $\mcal{V}_{\rm redo}^{(i+1)}$ as the union of: (i) the set of all replica vertices introduced in $\mcal{L}_\mcal{G}^{(i+1)}$, (ii) the set of descendant vertices of $v_{k_mk_{m+1}}$ in $\mcal{T}_i$, and (iii) the vertex $v_{k_mk_{m+1}}$. 
For any vertex $v \not\in \mcal{V}_{\rm redo}^{(i+1)}$, the path connecting $v_{SS'}$ to $v$ in $\mcal{T}_i$ does not pass through $\mcal{P}^{(i)}_{\rm to-fix}$. As a result, we can copy this part of $\mcal{T}_i$ to $\mcal{T}_{i+1}$ without loss of generality. Clearly, replica vertices never existed before $\mcal{L}_\mcal{G}^{(i+1)}$ so there is no known path for them in $\mcal{T}_i$. Similarly, the path from $v_{S'S}$ to $v_{k_mk_{m+1}}$  (and its descendants) passes through $\mcal{P}^{(i)}_{\rm to-fix}$ , thus we need to find a new route for them now that the chordal path has been eliminated from the graph.
Also it is not difficult to see that any $v' \not\in \mcal{V}^{(i+1)}_{\rm redo}$ will not be a descendant of $v$, $\forall v \in \mcal{V}^{(i+1)}_{\rm redo}$ as this would contradict the need to find a new path for some vertex in $\mcal{V}^{(i+1)}_{\rm redo}$.

   \begin{figure}[t]
        \centering
        \includegraphics[width=0.28\textwidth]{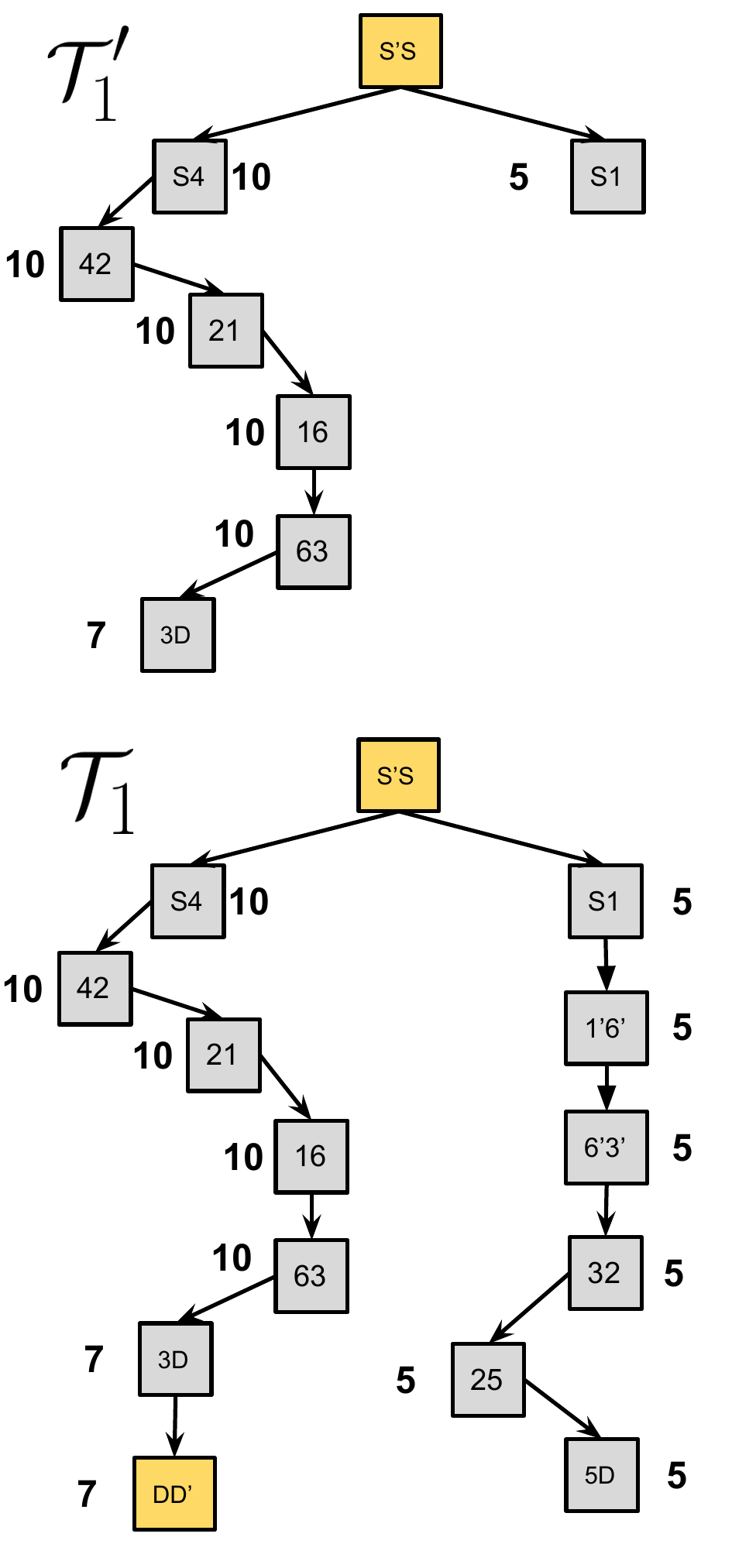}
        \caption{Spanning trees $\mcal{T}'_1$ and  $\mcal{T}_1$ for $\mcal{L}_\mcal{G}^{(1)}$ in Fig.~\ref{fig:line_graph_modified_step3}.
        Boldface numbers represent the FD capacity with which a node can be reached from $SS'$ using each tree.}
        \label{fig:line_graph_new_tree}
    \end{figure}

As per our discussion above, we find the rest of $\mcal{T}_{i+1}$ by initializing an intermediate point in the Dijkstra's algorithm and continue the execution of the algorithm from there. 
In particular, we start from the point where $\forall v \not\in \mcal{V}^{(i+1)}_{\rm redo}$ have been expanded (and thus appear in $\mcal{T}_{i+1}$ with the same path as in $\mcal{T}_i$).
We denote the intermediate version of $\mcal{T}_{i+1}$ at this point as $\mcal{T}'_{i+1}$, which is a pruned version of the tree $\mcal{T}_i$.
Note that, at any iteration of the classical Dijkstra's algorithm, a yet to be expanded vertex $v$ has a best so-far path from $v_{S'S}$ of FD capacity $c'(v)$. 
This achievable FD capacity {at an unexpanded vertex $v$} is based on {the maximum capacity achieved by each of the} neighbor vertices that have already been expanded and added to the spanning tree $\mcal{T}'_{i+1}$ {as well as the capacities of incident edges from those neighbor vertices to the vertex $v$.}
We denote the capacity of a neighbor vertex $v'$ that was already expanded as $\hat{c}^{\mcal{T}_{i+1}}_\mcal{L} (v')$.
We now note that the point from which we are going to start Dijkstra's algorithm is when the set of unexpanded vertices is $\mcal{V}_{\rm redo}^{(i+1)}$ and the vertices in $\mcal{T}'_{i+1}$ form the complement set ${\mcal{V}_{\rm redo}^{(i+1)}}^c$.
Thus, for the vertices still unexpanded (i.e., in $\mcal{V}^{(i+1)}_{\rm redo}$), the capacities currently achievable at them at this stage of the algorithm are initialized as
\[
    \widehat{c}_\mcal{L}(v) = \max_{v' \not\in \mcal{V}^{(i+1)}_{\rm redo}} \min\left \{ c^{\mcal{T}_i}_\mcal{L}(v),\ c_\mcal{L}(e_{v',v})\right \}.
\]
%
Now that we have the initialization of Dijkstra's algorithm to the state that we want, we run the standard routine of the algorithm to continue expanding the vertices in $\mcal{V}^{(i+1)}_{\rm redo}$.
When all the vertices have been expanded, we get the final tree $\mcal{T}_{i+1}$. 

\noindent{\bf Example.} For our running example and $i=0$, the tree $\mcal{T}'_{1}$ (which is a subset of $\mcal{T}_{0}$) and the new generated tree $\mcal{T}_1$ for $\mcal{L}_\mcal{G}^{(1)}$ are shown in Fig.~\ref{fig:line_graph_new_tree}.
It is worth noting that the spanning tree $\mathcal{T}_1$ in Fig.~\ref{fig:line_graph_new_tree} has the path $\mcal{P}^{(1)}_\mcal{L} = v_{SS'} - v_{S4} - v_{42} - v_{21} - v_{16} - v_{63} - v_{3D} - v_{DD'}$ of capacity $\msf{C}^{\rm FD}_{\mcal{P}^{(1)}_\mcal{L}} = 7$ that is chordless. Hence the algorithm returns this path and terminates (see Step~1).

%

\medskip

It is important to note that, from the replication procedure we do in Step 3, we add a number of replica vertices equal to the length of $\mcal{P}^{(i)}_{\rm to-fix}$ minus two (since we do not replicate the endpoints). Moreover, in addition to the replica vertices, only one endpoint of $\mcal{P}^{(i)}_{\rm to-fix}$ is a member of $\mcal{V}^{(i+1)}_{\rm redo}$ (i.e., the vertex $v_{k_mk_{m+1}}$).
As a result $\left | \mcal{V}^{(i)}_{\rm redo} \right | \leq \left|\mcal{V}({\mcal{L}_\mcal{G}})\right|, \forall i$.
{Thus, }the size of the network that Dijkstra's algorithm processes in Step 4 does not increase from one iteration to the next. 
This implies that Step 4 of the algorithm has a complexity that is at most $O(V_{\mcal{L}_\mcal{G}}\log V_{\mcal{L}_\mcal{G}} + E_{\mcal{L}_\mcal{G}})$ where $V_{\mcal{L}_\mcal{G}} = \left|\mcal{V}({\mcal{L}_\mcal{G}})\right|$ and $E_{\mcal{L}_\mcal{G}} = \left|\mcal{E}(\mcal{L}_\mcal{G})\right|$.
The time complexity of Steps 1, 2 and 3 is linear in $V_{\mcal{L}_\mcal{G}}$ and $E_{\mcal{L}_\mcal{G}}$.
Let ${K}_{\mcal{G}}$ be the number of cycles in $\mcal{G}$. {From the observation in Section~\ref{sec:chordal}, this} is equal to the number of chordal paths in $\mcal{L}_\mcal{G}$.
Since in each iteration over the four steps, we eliminate one chordal path, then for a line graph with $K_\mcal{G}$ chordal paths, we make at most $K_\mcal{G}$ iterations.
As a result, the complexity of the described algorithm for finding the simple chordless path with the largest FD capacity in $\mcal{L}_\mcal{G}$ is $O\left((K_\mcal{G}+1) (V_{\mcal{L}_\mcal{G}}\log V_{\mcal{L}_\mcal{G}} + E_{\mcal{L}_\mcal{G}})\right)$.

Note that the number of vertices in $\mcal{L}_\mcal{G}$ is equal to the number of edges in $\mcal{G}$ and the number of edges in $\mcal{L}_\mcal{G}$ is upper bounded by the number of edges in $\mcal{G}$ multiplied by the maximum vertex degree $d$. Additionally, the complexity of constructing a line digraph $\mcal{L}_\mcal{G}$ from a digraph $\mcal{G}$ is of order $O(|\mcal{E}({\mcal{G}})| d)$.
Thus, the problem of finding the simple path in $\mcal{G}$ with the largest HD approximate capacity is equivalent to creating the line digraph $\mcal{L}_\mcal{G}$ with FD capacities and then finding the chordless path with the largest FD capacity in that line digraph $\mcal{L}_\mcal{G}$.
The computational complexity of this procedure is 
$O\left(|\mcal{E}({\mcal{G}})| d + (K_\mcal{G}+1) (V_{\mcal{L}_\mcal{G}}\log V_{\mcal{L}_\mcal{G}} + E_{\mcal{L}_\mcal{G}})\right)$ = 
$O\left( (K_\mcal{G}+1) ( |\mcal{E}({\mcal{G}})| \log |\mcal{E}({\mcal{G}})|  + |\mcal{E}({\mcal{G}})| d)\right)$.

\begin{cor}\label{cor:complexity}
    {\rm If the number of cycles in a network with $N$ relays (described by the digraph $\mcal{G}$) is at most polynomial $O(N^\alpha)$, then we can find the simple path with the largest HD approximate capacity in polynomial-time, i.e., in $O((N^\alpha+1) ( |\mcal{E}({\mcal{G}})| \log |\mcal{E}({\mcal{G}})|  + |\mcal{E}({\mcal{G}})| d))$. This holds even when we do not have an a priori knowledge of the location of the cycles.}
\end{cor}

As a network example for which Corollary~\ref{cor:complexity} applies, we can study the layered network where the relays are arranged as $N_L$ relays per layer over $L$ layers. Every relay can only communicate with the relays in the following layer of relays. It is not difficult to see that for this particular network, the number of cycles in the graph is equal to zero, i.e., $K_\mcal{G}=0$. In addition, the maximum degree $d$ of a vertex is $O(N_L)$ and the number of edges in the network is $\Theta(L N_L^2)$. By substituting these values to the expression in Corollary~\ref{cor:complexity}, we get that the complexity of finding a simple path with the largest HD approximate capacity in a layered network is given by
\begin{align*}
    &O((K_\mcal{G}+1) ( |\mcal{E}({\mcal{G}})| \log |\mcal{E}({\mcal{G}})|  + |\mcal{E}({\mcal{G}})| d)) \\ &\quad= O( LN_L^2\log LN_L^2 + LN_L^2  N_L ) 
    \\ &\quad= O( LN_L^2\log L + 2 LN_L^2\log N_L + LN_L^3) \\
    &\quad= O( LN_L^2\log L + LN_L^3).
\end{align*}


\section{Conclusion}\label{sec:conc}

In this work we proved that, given a network with a source node, a destination node and a number of relays, finding the path from the source to the destination with the largest HD approximate capacity is NP-hard in general. 
This represents a surprising result and it is fundamentally different from the FD counterpart, since the path with the largest FD capacity can always be discovered in polynomial-time.
We also showed that, if the number of cycles inside the network is polynomial in the number of nodes, then a polynomial-time algorithm exists to find the best HD path.

Future work would include discovering alternative sufficient conditions that allow to develop polynomial-time algorithms (similar to the one on the number of cycles presented in this paper) and designing low-complexity algorithms that, even if not optimal, ensure that a significant portion of the best HD path approximate capacity can always be achieved.

\bibliographystyle{IEEEtran}
\bibliography{LineNetwork}

\begin{thebibliography}{10}
\providecommand{\url}[1]{#1}
\csname url@samestyle\endcsname
\providecommand{\newblock}{\relax}
\providecommand{\bibinfo}[2]{#2}
\providecommand{\BIBentrySTDinterwordspacing}{\spaceskip=0pt\relax}
\providecommand{\BIBentryALTinterwordstretchfactor}{4}
\providecommand{\BIBentryALTinterwordspacing}{\spaceskip=\fontdimen2\font plus
\BIBentryALTinterwordstretchfactor\fontdimen3\font minus
  \fontdimen4\font\relax}
\providecommand{\BIBforeignlanguage}[2]{{%
\expandafter\ifx\csname l@#1\endcsname\relax
\typeout{** WARNING: IEEEtran.bst: No hyphenation pattern has been}%
\typeout{** loaded for the language `#1'. Using the pattern for}%
\typeout{** the default language instead.}%
\else
\language=\csname l@#1\endcsname
\fi
#2}}
\providecommand{\BIBdecl}{\relax}
\BIBdecl

\bibitem{duarte2014design}
M.~Duarte, A.~Sabharwal, V.~Aggarwal, R.~Jana, K.~Ramakrishnan, C.~W. Rice, and
  N.~Shankaranarayanan, ``Design and characterization of a full-duplex
  multiantenna system for wifi networks,'' \emph{IEEE Transactions on Vehicular
  Technology}, vol.~63, no.~3, pp. 1160--1177, 2014.

\bibitem{everett2016softnull}
E.~Everett, C.~Shepard, L.~Zhong, and A.~Sabharwal, ``Softnull: Many-antenna
  full-duplex wireless via digital beamforming,'' \emph{IEEE Transactions on
  Wireless Communications}, vol.~15, no.~12, pp. 8077--8092, 2016.

\bibitem{wang2017primer}
Y.-P.~E. Wang, X.~Lin, A.~Adhikary, A.~Grovlen, Y.~Sui, Y.~Blankenship,
  J.~Bergman, and H.~S. Razaghi, ``A primer on 3gpp narrowband internet of
  things,'' \emph{IEEE Communications Magazine}, vol.~55, no.~3, pp. 117--123,
  2017.

\bibitem{awerbuch2004high}
B.~Awerbuch, D.~Holmer, and H.~Rubens, ``High throughput route selection in
  multi-rate ad hoc wireless networks,'' \emph{Wireless On-Demand Network
  Systems}, pp. 201--205, 2004.

\bibitem{de2005high}
D.~S. De~Couto, D.~Aguayo, J.~Bicket, and R.~Morris, ``A high-throughput path
  metric for multi-hop wireless routing,'' \emph{Wireless networks}, vol.~11,
  no.~4, pp. 419--434, 2005.

\bibitem{broch1998performance}
J.~Broch, D.~A. Maltz, D.~B. Johnson, Y.-C. Hu, and J.~Jetcheva, ``A
  performance comparison of multi-hop wireless ad hoc network routing
  protocols,'' in \emph{Proceedings of the 4th annual ACM/IEEE international
  conference on Mobile computing and networking}.\hskip 1em plus 0.5em minus
  0.4em\relax ACM, 1998, pp. 85--97.

\bibitem{ARXIV_version}
Y.~H. Ezzeldin, M.~Cardone, C.~Fragouli, and D.~Tuninetti, ``Efficiently
  finding simple half-duplex schedules in {G}aussian relay line networks,'' in
  \emph{IEEE International Symposium on Information Theory (ISIT)}, June 2017,
  pp. 471--475.

\bibitem{ARXIV_submodularity}
M.~Cardone, Y.~H. Ezzeldin, C.~Fragouli, and D.~Tuninetti, ``Network
  simplification in half-duplex: Building on submodularity,''
  \emph{ar{X}iv:1607.01441}, July 2017.

\bibitem{karp1972reducibility}
R.~M. Karp, ``Reducibility among combinatorial problems,'' in \emph{Complexity
  of computer computations}.\hskip 1em plus 0.5em minus 0.4em\relax Springer,
  1972, pp. 85--103.

\bibitem{aref1981information}
M.~R. Aref, ``Information flow in relay networks,'' \emph{Ph.D. thesis,
  Stanford University}, 1981.

\bibitem{OzgurIT2013}
A.~{\"O}zg{\"u}r and S.~N. Diggavi, ``Approximately achieving {G}aussian relay
  network capacity with lattice-based {QMF} codes,'' \emph{IEEE Transactions on
  Information Theory}, vol.~59, no.~12, pp. 8275--8294, December 2013.

\bibitem{KramerAllerton2004}
G.~Kramer, ``Models and theory for relay channels with receive constraints,''
  in \emph{42nd Annual Allerton Conference on Communication, Control, and
  Computing}, September 2004, pp. 1312--1321.

\bibitem{AODV}
C.~E. Perkins and E.~M. Royer, ``{Ad-hoc On-Demand Distance Vector Routing},''
  in \emph{Proceedings of the Second IEEE Workshop on Mobile Computer Systems
  and Applications}, 1999, p.~90.

\bibitem{OLSR}
T.~Clausen and P.~Jacquet, ``{Optimized Link State Routing Protocol
  ({OLSR})},'' \emph{RFC 3626, DOI 10.17487/RFC3626}, 2003.

\bibitem{DSR}
D.~B. Johnson and D.~A. Maltz, ``{Dynamic Source Routing in Ad Hoc Wireless
  Networks},'' in \emph{Mobile computing}.\hskip 1em plus 0.5em minus
  0.4em\relax Springer, 1996, pp. 153--181.

\bibitem{pollack1960letter}
M.~Pollack, ``Letter to the editor—the maximum capacity through a network,''
  \emph{Operations Research}, vol.~8, no.~5, pp. 733--736, 1960.

\bibitem{dijkstra1959note}
E.~W. Dijkstra, ``A note on two problems in connexion with graphs,''
  \emph{Numerische mathematik}, vol.~1, no.~1, pp. 269--271, 1959.

\bibitem{ahmed2009shortest}
M.~Ahmed and A.~Lubiw, ``Shortest paths avoiding forbidden subpaths,'' in
  \emph{26th International Symposium on Theoretical Aspects of Computer Science
  STACS 2009}.\hskip 1em plus 0.5em minus 0.4em\relax IBFI Schloss Dagstuhl,
  2009, pp. 63--74.

\end{thebibliography}

\end{document}